\documentclass[11pt]{article}
\usepackage[height=220mm,width=150mm]{geometry}
\usepackage{array}
\usepackage{color}
\usepackage{mathrsfs}
\usepackage{hyperref}
\usepackage{amssymb}
\usepackage{amsmath}
\usepackage{amsthm}
\allowdisplaybreaks
\usepackage{graphicx}
\usepackage{tikz}
\usetikzlibrary{shapes.arrows,chains,positioning}
\newtheorem{Theorem}{Theorem}[section]
\newtheorem{Lemma}[Theorem]{Lemma}
\newtheorem{Definition}[Theorem]{Definition}

\newtheorem{Corollary}[Theorem]{Corollary}
\newtheorem{Remark}[Theorem]{Remark}

\begin{document}
%

\title{Sparse Phase Retrieval with Redundant Dictionary via $\ell_q (0<q\le 1)$-Analysis Model\thanks{This work was supported in part by the National Natural Science Foundation of China under Grants 12261059 and 62202442, in part by the Jiangxi Provincial Natural Science Foundation under Grant 20224BAB211001, and in part by the Anhui Provincial Natural Science Foundation under Grant 2208085QF188.}}

\author{Haiye Huo$^a$ and Li Xiao$^{b,}$\thanks{Corresponding author.}\\
$^a$ Department of Mathematics, School of Mathematics and Computer Sciences, \\ Nanchang University, Nanchang~330031, Jiangxi, China\\
\mbox{} \\
$^b$ Department of Electronic Engineering and Information Science, University of \\Science and Technology of China, Hefei~230052, China\\
\normalsize{Emails: hyhuo@ncu.edu.cn; xiaoli11@ustc.edu.cn}
}

\date{}
\maketitle

\textbf{Abstract:}\,\,
Sparse phase retrieval with redundant dictionary is to reconstruct the signals of interest that are (nearly) sparse in a redundant dictionary or
frame from the phaseless measurements via the optimization models. Gao \cite{Gao2017} presented conditions on the measurement matrix, called null space property (NSP) and strong dictionary restricted isometry property (S-DRIP), for exact and stable recovery of dictionary-$k$-sparse signals via the $\ell_1$-analysis model for sparse phase retrieval with redundant dictionary, respectively, where, in particularly, the S-DRIP of order $tk$ with $t>1$ was derived. In this paper, motivated by many advantages of the $\ell_q$ minimization with $0<q\leq1$, e.g., reduction of the number of measurements required, we generalize these two conditions to the $\ell_q$-analysis model. Specifically, we first present two NSP variants for exact recovery of dictionary-$k$-sparse signals via the $\ell_q$-analysis model in the noiseless scenario. Moreover, we investigate the S-DRIP of order $tk$ with $0<t<\frac{4}{3}$ for stable recovery of dictionary-$k$-sparse signals via the $\ell_q$-analysis model in the noisy scenario, which will complement the existing result of the S-DRIP of order $tk$ with $t\geq2$ obtained in \cite{CH22}.


\section{Introduction}\label{sec:S0}
The problem of recovering a signal from phaseless measurements is known as phase retrieval \cite{Walter1963}. It has emerged in quantum mechanics \cite{LR2009}, optics \cite{SEC2015}, and X-ray crystallography \cite{Mill90}, etc. In general form, the problem of phase retrieval is to reconstruct an original unknown signal $x_{0}\in \mathbb{H}^{n} (\mathbb{H}=\mathbb{R}$ or $\mathbb{C})$ from its magnitude only measurements
\begin{equation}\label{model0}
|Ax|=|Ax_0|+e,
\end{equation}
where $A=[a_1,a_2,\cdots,a_m]^{*}\in\mathbb{H}^{m\times n}$ denotes a measurement matrix with $Ax_0=[\langle a_1, x_0\rangle,\\
\langle a_2,x_0\rangle,\cdots,\langle a_m, x_0\rangle]^{*}$, $e\in \mathbb{H}^{m}$ denotes the noise, and $*$ and $\lvert\cdot\rvert$ denote the conjugate transpose and the absolute value operations, respectively. It is worth noting that the phase retrieval setup (\ref{model0}) will inevitably lead to ambiguous solutions; that is, the global phase shift, conjugate inversion, and spatial shift of any solution to (\ref{model0}) are solutions as well. Therefore, these trivial ambiguities are considered acceptable and signal recovery in phase retrieval is referred to as reconstructing a signal up to a unimodular constant (or a global phase).

With the development of compressed sensing \cite{HSX2018,Sun2010,Wan20,Xie23}, a great number of studies have focused on the so-called sparse phase retrieval by incorporating signal sparsity as prior information into the standard phase retrieval problem \cite{Huo22,ZSZ22,ZY20}. Specifically, a variety of sufficient and/or necessary conditions on the measurement matrix for exact and stable recovery of $k$-sparse signals in compressed sensing have been extended to the sparse phase retrieval problem \cite{CJL22,CH22,YHW17}. For example, Wang and Xu \cite{WX2014} introduced a sufficient and necessary condition (namely the null space property (NSP)) for exact recovery of sparse signals from their magnitude only measurements. The restricted isometry property (RIP) was correspondingly generalized as the strong RIP (S-RIP), which is a sufficient condition to ensure stable recovery of a sparse signal in sparse phase retrieval \cite{VX2016,GWX2016}. We refer the reader to \cite{GWX2016,WX2014,VX2016} for more details on sparse phase retrieval.

Recently, considering that the signals of interest (e.g., radar images \cite{HS2009}) are generally not sparse in the standard orthonormal basis but have sparse representations in a redundant dictionary or frame, many researchers have began to pay attention to sparse phase retrieval for such dictionary-sparse signals \cite{CL16}. Specifically, let $D\in \mathbb{H}^{n\times N}$ with $n<N$ be a redundant dictionary. Assume that an unknown signal $x_0\in \mathbb{H}^n$ is sparse in the redundant dictionary $D$, i.e., there exists a sparse vector $z_0\in \mathbb{H}^N$ such that $x_0=Dz_0$. The problem of recovering a dictionary-sparse signal $x_0$ from the noisy measurements
\begin{equation}\label{Dx}
|ADz|=|ADz_0|+e
\end{equation}
is called sparse phase retrieval with redundant dictionary. Without loss of generality, the redundant dictionary $D$ is supposed to be a tight frame. A signal $z\in \mathbb{H}^N$ is said to be $k$-sparse, if $\|z\|_{0}\le k$, i.e., there exist at most $k$ nonzero entries in $z$. Let
\begin{equation}\label{xin1}
\mathbb{H}_{k}^{N}:=\{z\in \mathbb{H}^N:\;\|z\|_{0}\le k\}
\end{equation}
be the set of $k$-sparse signals in $\mathbb{H}^N$. Accordingly, define
\begin{equation}\label{xin2}
D\mathbb{H}_{k}^{N}:=\{x\in \mathbb{H}^n:\; \mbox{there exists one}\; z\in\mathbb{H}_{k}^{N} \mbox{such that}\; x=Dz\}
\end{equation}
as the set of dictionary-$k$-sparse signals (with respect to $D$) in $\mathbb{H}^N$. The $\ell_1$-analysis model was proposed in \cite{Gao2017} to reconstruct a dictionary-$k$-sparse $x_0$, i.e.,
\begin{equation}\label{mod:l1}
\min\|D^*x\|_{1} \quad{\mbox{subject to}}\quad\||Ax|-|Ax_0|\|_{2}^2\le\epsilon^2,
\end{equation}
where $\epsilon$ stands for the upper bound for the noise $e$. For the noiseless scenario (i.e., $e=0$ or $\epsilon=0$), Gao \cite{Gao2017} analyzed the null space of the measurement matrix and presented two novel necessary and sufficient conditions (i.e., the NSP) for exact recovery of real and complex dictionary-$k$-sparse signals via the $\ell_1$-analysis model (\ref{mod:l1}), respectively. Furthermore, for the noisy scenario, the strong dictionary RIP (S-DRIP) of order $tk$ with $t>1$ was provided to be a sufficient condition for stable recovery of real dictionary-$k$-sparse signals via the $\ell_1$-analysis model (\ref{mod:l1}) in \cite{Gao2017}. More recently, motivated by the advantages of the non-convex $\ell_q$-minimization with $q\in(0,1]$ \cite{Lai2011}, e.g., reduction of the number of measurements required, Cao and Huang \cite{CH22} investigated the $\ell_q$-analysis model, i.e.,
\begin{equation}\label{mod:lq}
\min\|D^*x\|_{q}^q \quad{\mbox{subject to}}\quad\||Ax|-|Ax_0|\|_{2}^2\le\epsilon^2,
\end{equation}
and derived the S-DRIP of order $tk$ with $t\geq2$ to guarantee stable recovery of real dictionary-$k$-sparse signals via the $\ell_q$-analysis model (\ref{mod:lq}).

To complement the results obtained in \cite{Gao2017,CH22} for sparse phase retrieval with redundant dictionary, we first extend the NSP as necessary and sufficient conditions under which (real and complex) dictionary-$k$-sparse signals can be accurately recovered via the $\ell_q$-analysis model (\ref{mod:lq}) in the noiseless scenario. In addition, we introduce the S-DRIP of order $tk$ with $0<t<\frac{4}{3}$ for stable recovery of real dictionary-$k$-sparse signals via the $\ell_q$-analysis model (\ref{mod:lq}) in the noisy scenario.

The rest of this paper is organized as follows. In Section~\ref{sec:S2}, we give some notations and recall some results which are quite useful for later proofs. In Sections~\ref{sec:S3} and \ref{sec:S4}, we present the NSP conditions for exact recovery of dictionary-$k$-sparse signals in the noiseless scenario, and the S-DRIP of order $tk$ with $0<t<\frac{4}{3}$ for stable recovery of real dictionary-$k$-sparse signals in the noisy scenario, via the $\ell_q$-analysis model (\ref{mod:lq}), respectively. Finally, we conclude this paper in Section~\ref{sec:S6}.

\section{Preliminaries}\label{sec:S2}
A set of vectors $\{d_i\in \mathbb{H}^n:i=1,2,\cdots,N\}$ is said to be a frame for $\mathbb{H}^n$, if there exist two constants $0<s\le t<\infty$ such that for all $f\in \mathbb{H}^n$,
\begin{equation}\label{frame}
s\|f\|_2^2\le\sum_{i=1}^{N}|\langle f,d_i\rangle|^2\le t\|f\|_2^2.
\end{equation}
If $s=t$, the frame is called a tight frame. Throughout this paper, without loss of generality, assume that $D=[d_1,d_2,\cdots,d_N]\in\mathbb{H}^{n\times N}$ is a tight frame obeying the Parseval relations
\begin{equation}
f=\sum_{j}\langle f,d_i\rangle d_i \;\mbox{ and }\; \|f\|_2^2=\sum_{i}|\langle f,d_i\rangle|^2,
\end{equation}
where $\{d_i\}_{i=1}^N\in \mathbb{H}^n$ are the $N$ columns of $D$. As such, it can be easily proved that for all $f\in \mathbb{H}^n$,
\begin{equation}
DD^*=I_n \;\mbox{ and }\; \|f\|_2^2=\|D^{*}f\|_2^2,
\end{equation}
where $I_n$ stands for the identity matrix of size $n\times n$.
Let $D^{\bot}$ denote the orthonormal complement of $D$. It is known that $D^{\bot}$ is a tight frame as well, and for all $z\in \mathbb{H}^N$, we have
\begin{equation}\label{xin0}
\|z\|_2^2=\|Dz\|_2^2+\|D^{\bot}z\|_2^2.
\end{equation}
For clarity of definitions in (\ref{xin1}), we use $\mathbb{R}_{k}^{N}$ and $\mathbb{C}_{k}^{N}$ to denote the set of $k$-sparse real signals in the real case (i.e., $\mathbb{H}=\mathbb{R}$) and the set of $k$-sparse complex signals in the complex case (i.e., $\mathbb{H}=\mathbb{C}$), respectively. Analogously for (\ref{xin2}), both $D\mathbb{R}_{k}^{N}$ and $D\mathbb{C}_{k}^{N}$ can be correspondingly defined.
The best $k$-term approximation error (with respect to the $\ell_q$ norm) of a vector $z\in \mathbb{H}^N$ is denoted by
\begin{equation}\label{xin3}
\sigma_{k}(z)_q:=\inf_{z^{\prime}\in\mathbb{H}_{k}^{N}}\|z^{\prime}-z\|_{q}.
\end{equation}

Define $[1:m]:=\{1,2,\cdots,m\}$. Let $T$ be a subset of $[1:m]$ with $|T|$ denoting its cardinal number, and $T^{c}:=[1:m]\backslash T$ be the complement of $T$. Let $A_{T}:=[a_i,\; i\in T]^*\in\mathbb{H}^{|T|\times n}$ be a sub-matrix of $A=[a_1,a_2,\cdots,a_m]^{*}\in\mathbb{H}^{m\times n}$, whose rows are the rows of $A$ with indices in $T$. We next recall the definitions of D-RIP \cite{CEN11} and S-DRIP \cite{Gao2017} in the real case (i.e., $\mathbb{H}=\mathbb{R}$), respectively, as follows.

\begin{Definition}\cite[D-RIP]{CEN11}\label{Def:DRIP}
For a tight frame $D\in\mathbb{R}^{n\times N}$ and a measurement matrix $A\in \mathbb{R}^{m\times n}$, we say that $A$ satisfies the restricted isometry property adapted to $D$ (D-RIP for short) of order $k$ with a constant $\delta_k$, if
\begin{equation}\label{Def:DRIP:1}
(1-\delta_k)\|Dz\|_{2}^2\le\|ADz\|_2^2\le(1+\delta_k)\|Dz\|_2^2
\end{equation}
holds for all $k$-sparse vectors $z\in \mathbb{R}_{k}^{N}$.
\end{Definition}

Note that if the measurement matrix $A$ satisfies the D-RIP in Definition \ref{Def:DRIP}, according to (\ref{xin0}), we readily have, for all $z\in\mathbb{R}_{k}^N$,
\begin{equation}\label{repl}
(1-\delta_k)\|z\|_2^2\le\|ADz\|_2^2+\|D^{\bot}z\|_2^2\le(1+\delta_k)\|z\|_2^2.
\end{equation}

\begin{Definition}\cite[S-DRIP]{Gao2017}\label{Def:SDRIP}
For a tight frame $D\in\mathbb{R}^{n\times N}$ and a measurement matrix $A\in \mathbb{R}^{m\times n}$, we say that $A$ satisfies the strong restricted isometry property adapted to $D$ (S-DRIP for short) of order $k$ with constants $\theta_{-},\;\theta_{+}\in (0,2)$, if
\begin{align}\label{Def:SDRIP:1}
\theta_{-}\|Dv\|_{2}^2&\le\min_{I\subseteq[1:m],\;|I|\geq \frac{m}{2}}\|A_IDv\|_2^2\le\max_{I\subseteq[1:m],\;|I|\geq \frac{m}{2}}\|A_{I}Dv\|_2^2\le\theta_{+}\|Dv\|_2^2
\end{align}
holds for all $k$-sparse vectors $v\in\mathbb{R}^{N}$.
\end{Definition}

Finally, we introduce some important lemmas, which will be needed later.

\begin{Lemma}\cite[Lemma 1]{ZL18}\label{lem1}
Let $T$ be an index set with $|T|=k$. Given vectors $\{v_i:i\in T\}$, choosing all subsets $T_i\subset T$ satisfying $|T_i|=l,i\in I$ with $|I|=\left(
                                                                                                                                                                                                                    \begin{array}{c}
                                                                                                                                                                                                                      k \\
                                                                                                                                                                                                                      l \\
                                                                                                                                                                                                                    \end{array}
                                                                                                                                                                                                                  \right)
$, we have
\begin{equation}\label{lem1:1}
\sum_{i\in I}\sum_{p\in T_i}v_p=\left(\begin{array}{c}
                                                                                                                                                                                                                      k-1 \\
                                                                                                                                                                                                                      l-1 \\
                                                                                                                                                                                                                    \end{array}
                                                                                                                                                                                                                  \right)
\sum_{p\in T}v_p\;\;\;(l\ge 1),
\end{equation}
and
\begin{equation}\label{lem1:2}
\sum_{i\in I}\sum_{p\ne q\in T_i}\langle v_p,v_q\rangle=\left(\begin{array}{c}
                                                                                                                                                                                                                      k-2 \\
                                                                                                                                                                                                                      l-2 \\
                                                                                                                                                                                                                    \end{array}
                                                                                                                                                                                                                  \right)
\sum_{p\ne q}\langle v_p,v_q\rangle\;\;\;(l\ge 2).
\end{equation}

\end{Lemma}

\begin{Lemma}\cite[Lemma 2.2]{ZL19}\label{lem2}
For any $x$ in the set
\[
V=\{x\in \mathbb{R}^n:\|x\|_{0}=r,\;\|x\|_q^q\le k\alpha^q,\;\|x\|_{\infty}\le\alpha\},
\]
where $k\le r$ is a positive integer, $\alpha$ is a positive constant, and $0<q\le 1$,
it can be represented by a convex combination of $k$-sparse vectors, i.e.,
\[
x=\sum_{i}\lambda_iu_i,
\]
where $\lambda_i>0,\;\sum_{i}\lambda_i=1,$ and $\|u_i\|_0\le k$. In addition,
\[
\sum_{i}\lambda_i\|u_i\|_2^2\le\min\left\{\frac{r}{k}\|x\|_2^2,\alpha^q\|x\|_{2-q}^{2-q}\right\}.
\]
\end{Lemma}

\begin{Lemma}\cite[Lemma 5.3]{CZ13}\label{Lem:Add1}
Let $r\ge k,$\;$b_1\ge b_2\ge\cdots\ge b_r\ge 0,$\;$d\ge 0$, and $\sum_{i=1}^{k}b_i+d\ge\sum_{i=k+1}^{r}b_i$. Then, for all $\omega\ge 1$,
\begin{equation}
\sum_{i=k+1}^{r}b_i^{\omega}\le k\left[\left(\frac{1}{k}\sum_{i=1}^{k}b_i^{\omega}\right)^{\frac{1}{\omega}}+\frac{d}{k}\right]^{\omega}.
\end{equation}
\end{Lemma}

\section{NSP for Exact Recovery of Dictionary-Sparse Signals}\label{sec:S3}
In this section, for any $x_0\in D\mathbb{H}_{k}^N$, we consider the model (\ref{mod:lq}) with $\epsilon=0$ (namely in the noiseless scenario), i.e.,
\begin{equation}\label{mod:lqN}
\min\|D^*x\|_{q}^q,\quad \mbox{subject to}\quad |Ax|=|Ax_0|.
\end{equation}
In what follows, we investigate two NSP variants, which are necessary and sufficient conditions to guarantee the exact recovery of real and complex dictionary-sparse signals in (\ref{mod:lqN}) up to a global phase, respectively.

\subsection{The Real Case}
We first consider the sparse phase retrieval with redundant dictionary in the real setting, where the signal, the measurement matrix, and the redundant dictionary are all in the real number filed.

\begin{Theorem}\label{Thm:PRN}
Given a measurement matrix $A\in \mathbb{R}^{m\times n}$ and a redundant dictionary $D\in\mathbb{R}^{n\times N}$, the following properties are equivalent:
\begin{enumerate}
  \item[(a)] For any $x_0\in D\mathbb{R}_{k}^{N}$, we have
  \[
  \mathop{\arg\min}_{x\in \mathbb{R}^{n}}\{\|D^*x\|_{q}^q: |Ax|=|Ax_0|\}=\{\pm x_0\}.
  \]
  \item[(b)] For every $\Lambda\subseteq[1:m]$ with $|\Lambda|\le k$, it holds
  \[
  \|D^*(u+v)\|_{q}^q<\|D^*(u-v)\|_{q}^q
  \]
  for all nonzero $u\in \mathcal{N}(A_{\Lambda})$ and $v\in \mathcal{N}(A_{\Lambda^c})$ satisfying $u+v\in D\mathbb{R}_{k}^N$, where $\mathcal{N}(A)$ denotes the null space of $A$.
\end{enumerate}
\end{Theorem}

\begin{proof}
The proof is similar to that of Theorem 3.1 in \cite{Gao2017}. We first prove $(b)\Rightarrow(a)$. Suppose that $(a)$ is false, i.e., there exists a solution $\hat{x}\ne\pm x_0$ to (\ref{mod:lqN}). Thus, we obtain
\begin{equation}\label{Thm:PRN0}
|A\hat{x}|=|Ax_0|,
\end{equation}
and
\begin{equation}\label{Thm:PRN3}
\|D^*\hat{x}\|_{q}^q\le\|D^*x_0\|_{q}^q.
\end{equation}
Let $\left\{a_i^T,\;i=1,2,\cdots,m\right\}$ be all the rows of $A$. Then, (\ref{Thm:PRN0}) implies that there exists a subset $\Lambda\subseteq[1:m]$ satisfying
\begin{equation*}\label{Thm:PRN4}
\langle a_i,x_0+\hat{x}\rangle=0,\;\; {\rm{for}}\;\; i\in \Lambda,
\end{equation*}
and
\begin{equation*}\label{Thm:PRN5}
\langle a_i,x_0-\hat{x}\rangle=0,\;\; {\rm{for}}\;\; i\in \Lambda^c,
\end{equation*}
namely,
\[
A_{\Lambda}(x_0+\hat{x})=0,\;A_{\Lambda^c}(x_0-\hat{x})=0.
\]
Set
\[
u:=x_0+\hat{x}\;\text{ and } \;v:=x_0-\hat{x}.
\]
Since $\hat{x}\ne\pm x_0$, we get $u\in \mathcal{N}(A_{\Lambda})\backslash \{0\}$,\;$v\in \mathcal{N}(A_{\Lambda^c})\backslash \{0\}$, and $u+v=2x_0\in D\mathbb{R}_{k}^N$.
Then, by $(b)$, we have
\[
\|D^*x_0\|_{q}^q<\|D^*\hat{x}\|_{q}^q,
\]
which is contradicted with (\ref{Thm:PRN3}).

Next, we prove $(a)\Rightarrow(b)$. Suppose that $(b)$ is false. In other words, there exist a subset $\Lambda\subseteq[1:m]$ with $|\Lambda|\le k$, and nonzero $u\in \mathcal{N}(A_\Lambda)$ and $v\in \mathcal{N}(A_{\Lambda^c})$ satisfying
\[
\|D^*(u+v)\|_{q}^q\ge\|D^*(u-v)\|_{q}^q,
\]
with $u+v\in D\mathbb{R}_{k}^N$.
Set $x_0:=u+v\in D\mathbb{R}_{k}^N$ and $\hat{x}:=u-v$. Then, we observe that $\hat{x}\neq\pm x_0$, and
\begin{equation}\label{Thm:PRN1}
\|D^*\hat{x}\|_{q}^q\le\|D^*x_0\|_{q}^q.
\end{equation}
Let $\left\{a_i^T,\;i=1,2,\cdots,m\right\}$ be all the rows of $A$. Then, we have
\[
2\langle a_i,u\rangle=\langle a_i,x_0+\hat{x}\rangle, \quad{\rm{for}}\; i=1,2,\cdots,m,
\]
and
\[
2\langle a_i,v\rangle=\langle a_i,x_0-\hat{x}\rangle, \quad {\rm{for}}\;i=1,2,\cdots,m.
\]
Due to $u\in \mathcal{N}(A_\Lambda)\backslash \{0\}$ and $v\in \mathcal{N}(A_{\Lambda^c})\backslash \{0\}$, we get
\[
\langle a_i,x_0\rangle=-\langle a_i,\hat{x}\rangle\;\; {\rm{for}}\;\; i\in \Lambda,
\]
and
\[
\langle a_i,x_0\rangle=\langle a_i,\hat{x}\rangle\;\; {\rm{for}}\;\; i\in \Lambda^c,
\]
which yields
\begin{equation}\label{Thm:PRN2}
|Ax_0|=|A\hat{x}|.
\end{equation}
Combining (\ref{Thm:PRN1}) and (\ref{Thm:PRN2}), one can see that $\hat{x}$ is a solution to (\ref{mod:lqN}), but $\hat{x}\neq\pm x_0$. This leads to a contradiction.
\end{proof}

\subsection{The Complex Case}
We next consider the sparse phase retrieval with redundant dictionary in the complex setting, where the signal, the measurement matrix, and the redundant dictionary are all in the complex number filed.
We call $\Omega=\{\Omega_1,\Omega_2,\cdots,\Omega_p\}$ a partition of $[1:m]$, if
\[
\Omega_i\subseteq[1:m], \; \bigcup_{i=1}^p\Omega_i=[1:m],\; \Omega_i\cap \Omega_j=\emptyset {\mbox{ for all }}i\ne j.
\]
Denote $\mathbb{S}:=\{c\in \mathbb{C}:|c|=1\}$. We have the following theorem, which is an extension of Theorem \ref{Thm:PRN} in the complex case.

\begin{Theorem}\label{Thm:Comp}
For a measurement matrix $A\in \mathbb{C}^{m\times n}$ and a  redundant dictionary $D\in\mathbb{C}^{n\times N}$, the following statements are equivalent:
\begin{itemize}
  \item [(a)] For any $x_0\in D\mathbb{C}_{k}^{N}$, we have
  \begin{equation*}\label{Thm:C0}
  \mathop{\arg\min}_{x\in \mathbb{C}^n}\{\|D^*x\|_{q}^q:|Ax|=|Ax_0|\}=\{cx_0:\;c\in \mathbb{S}\}.
  \end{equation*}
  \item [(b)] Assume that $\Omega=\{\Omega_1,\Omega_2,\cdots,\Omega_p\}$ is a partition of $[1:m]$, and that nonzero
  $\left\{\varphi_{i}\in \mathcal{N}(A_{\Omega_i})\right\}_{i=1}^p$ satisfy
  \begin{eqnarray}\label{Thm:C1}
  \frac{\varphi_1-\varphi_j}{d_1-d_j}=\frac{\varphi_1-\varphi_i}{d_1-d_i}\in D\mathbb{C}_{k}^{N}\backslash\{0\}\;\mbox{ for all } i,\;j\in[2:p] \text{ with }i\ne j,
  \end{eqnarray}
  for some distinct $d_1,d_2,\cdots,d_p\in \mathbb{S}$. Then,
  \begin{equation}\label{Thm:C2}
  \|D^*(\varphi_i-\varphi_j)\|_{q}^q<\|D^*(d_j\varphi_i-d_i\varphi_j)\|_{q}^q,
  \end{equation}
  holds for all $i,\;j\in[1:p]$ with $i\ne j$.
\end{itemize}
\end{Theorem}

\begin{proof}
The proof is similar to that of Theorem 3.2 in \cite{Gao2017}. We first prove $(b)\Rightarrow(a)$. Suppose that (a) is false, i.e., there exists a solution $\hat{x}\notin\{cx_0:\, c\in \mathbb{S}\}$ to (\ref{mod:lqN}) such that
\begin{equation}\label{Thm:C3}
\|D^*\hat{x}\|_{q}^q\le\|D^*x_0\|_{q}^q,
\end{equation}
and
\begin{equation}\label{Thm:C4}
|A\hat{x}|=|Ax_0|.
\end{equation}
Define all the rows of $A$ by $\left\{a_i^*,\; i=1,2,\cdots,m\right\}$. Then, we can derive
\begin{equation}\label{Thm:C5}
\langle a_i,\hat{x}\rangle=\langle a_i,d_ix_0\rangle,
\end{equation}
where $d_i\in \mathbb{S},\,i=1,2,\cdots,m$. We can use $\sim $ to represent an equivalence relation on $[1:m]$; that is to say, $i\sim j$, for $d_i=d_j$. Thus,
this equivalence relation can result in a partition $\Omega=\{\Omega_1,\Omega_2,\cdots,\Omega_p\}$ of $[1:m]$. For any $\Omega_i$, we know
\[
A_{\Omega_i}\hat{x}=A_{\Omega_i}(d_ix_0), \;i=1,2,\cdots,p.
\]
Set
\begin{equation}\label{Thm:C6}
\varphi_i:=d_ix_0-\hat{x},\;\; i=1,2,\cdots,p.
\end{equation}
We know $\varphi_i\in \mathcal{N}(A_{\Omega_i})\backslash\{0\}$, and
\[
\frac{\varphi_1-\varphi_j}{d_1-d_j}=\frac{\varphi_1-\varphi_i}{d_1-d_i}=x_0\in D\mathbb{C}_{k}^{N},\;\mbox{ for all } i,\;j\in[2:p],\; i\ne j.
\]
From the condition $(b)$, we have
\[
\|D^*(d_i-d_j)x_0\|_{q}^q<\|D^*(d_i-d_j)\hat{x}\|_{q}^q,
\]
namely,
\[
\|D^*x_0\|_{q}^q<\|D^*\hat{x}\|_{q}^q,
\]
which is contradicted with (\ref{Thm:C3}).

Next, we prove $(a)\Rightarrow(b)$. Suppose that $(b)$ is false. In other words, there exist a partition $\Omega=\{\Omega_1,\Omega_2,\cdots,\Omega_p\}$ of $[1:m]$, and nonzero $\left\{\varphi_i\in \mathcal{N}(A_{\Omega_i}),\;i\in [1:p]\right\}$ satisfying (\ref{Thm:C1}), but
\begin{equation*}\label{Thm:C8}
\|D^*(\varphi_{i_0}-\varphi_{j_0})\|_{q}^q\ge\|D^*(d_{j_0}\varphi_{i_0}-d_{i_0}\varphi_{j_0})\|_{q}^q
\end{equation*}
for some distinct $i_0,\;j_0\in[1:p]$.
Let
\begin{equation}\label{Thm:C10}
x_0:=\varphi_{i_0}-\varphi_{j_0}\in D\mathbb{C}_{k}^N\backslash\{0\},
\end{equation}
and
\begin{equation}\label{Thm:C9}
\hat{x}:=d_{j_0}\varphi_{i_0}-d_{i_0}\varphi_{j_0},\; d_{j_0}\ne d_{i_0}.
\end{equation}
Therefore,
\[
\hat{x}\notin \{cx_0, \;c\in \mathbb{S}\},
\]
and
\begin{equation}\label{Thm:add}
\|D^*\hat{x}\|_{q}^q\le\|D^*x_0\|_{q}^q.
\end{equation}
Define all the rows of $A$ by $\left\{a_i^*,\; i=1,2,\cdots,m\right\}$. Since $\varphi_i\in \mathcal{N}(A_{\Omega_i})\backslash\{0\}$, we have
\[
\langle a_k,\varphi_{i_0}\rangle=0\;\; \mbox{or} \,\;\langle a_k,\varphi_{j_0}\rangle=0,\;\; k\in \Omega_{i_0}\cup \Omega_{j_0}.
\]
Hence, from the definitions (\ref{Thm:C10})and (\ref{Thm:C9}), we obtain
\begin{equation}\label{Thm:C11}
|\langle a_k,x_0\rangle|=|\langle a_k,\hat{x}\rangle|,\;\; k\in \Omega_{i_0}\cup \Omega_{j_0}.
\end{equation}
For $k\notin \Omega_{i_0}\cup \Omega_{j_0}$, we assume that $k\in \Omega_{s}\;(s\ne i_0,\;j_0)$, namely, $\langle a_k,\varphi_{s}\rangle=0$.
From (\ref{Thm:C1}), we get
\begin{equation}\label{Thm:C12}
\frac{\varphi_i-\varphi_j}{d_i-d_j}=\frac{\varphi_l-\varphi_t}{d_l-d_t},
\end{equation}
for all $i,j,l,t\in [1:p]$ with $i\ne j$ and $l\ne t$.
Define
\[
y_0:=\frac{\varphi_{i_0}-\varphi_s}{d_{i_0}-d_s}=\frac{\varphi_{j_0}-\varphi_s}{d_{j_0}-d_s}\in D\mathbb{C}_{k}^N\backslash\{0\}.
\]
Then,
\[
\varphi_{i_0}=(d_{i_0}-d_s)y_0+\varphi_s,
\]
and
\[
\varphi_{j_0}=(d_{j_0}-d_s)y_0+\varphi_s.
\]
As a result, $x_0$ and $\hat{x}$ can be represented by
\[
x_0=\varphi_{i_0}-\varphi_{j_0}=(d_{i_0}-d_{j_0})y_0,
\]
and
\[
\hat{x}=d_{j_0}\varphi_{i_0}-d_{i_0}\varphi_{j_0}=(d_{i_0}-d_{j_0})d_sy_0+(d_{j_0}-d_{i_0})\varphi_s,
\]
respectively. As $\langle a_k,\varphi_s\rangle=0$, we have
\[
|\langle a_k, \hat{x}\rangle|=|\langle a_k,x_0\rangle|,\; k\in \Omega_{s}.
\]
We can prove that the claim holds for the other subsets $\Omega_i$ by using a similar argument. Thus,
\begin{equation}\label{Thm:C13}
|\langle a_k, \hat{x}\rangle|=|\langle a_k,x_0\rangle|\;\;\mbox{for all}\;\;k.
\end{equation}
(\ref{Thm:add}) and (\ref{Thm:C13}) implies that $\hat{x}\notin\{cx_0:\;c\in \mathbb{S}\}$ is a solution to (\ref{mod:lqN}), which leads to a contradiction. This completes the proof.
\end{proof}

\begin{Remark}
When $q=1$, the $\ell_q$-analysis model (\ref{mod:lqN}) reduces to the $\ell_1$-analysis model for sparse phase retrieval with redundant dictionary studied in \cite{Gao2017}, and accordingly our Theorems \ref{Thm:PRN} and \ref{Thm:Comp} above reduce to Theorems 3.1 and 3.2 in \cite{Gao2017}, respectively.
\end{Remark}

To sum up, two novel NSP conditions we proposed in Theorems \ref{Thm:PRN} and \ref{Thm:Comp} above can guarantee the uniqueness (up to a global phase) of sparse phase retrieval with redundant dictionary via the $\ell_q$ minimization (\ref{mod:lqN}) for both the real and complex settings in the noiseless scenario, respectively.

\section{S-DRIP for Stable Recovery of Real Dictionary-Sparse Signals}\label{sec:S4}
In this section, a novel S-DRIP condition for stable recovery of a dictionary-sparse signal from the phaseless measurements via the $\ell_q$-analysis model (\ref{mod:lq}) with noise is studied in the real setting. Specifically, the S-DRIP of order $tk$ with $t\geq2$ was presented in \cite{CH22} with respect to the $\ell_q$-analysis model (\ref{mod:lq}), but in this paper we consider the S-DRIP of order $tk$ with $0<t<\frac{4}{3}$ in order to broaden the scope of conditions for stable recovery of real dictionary-$k$-sparse signals.

\begin{Theorem}\label{Thm:S-DRIP}
Let $D\in \mathbb{R}^{n\times N}$ be a tight frame, and $x_0\in D\mathbb{R}_{k}^{N}$. Suppose that a measurement matrix $A\in \mathbb{R}^{m\times n}$ satisfies the S-DRIP of order $tk$ with constants $\theta_{-},\;\theta_{+}\in(0,2)$ satisfying
\[
\max\left\{\frac{(3+2^{\frac{2}{q}-2})(1-\theta_{-})}{2-\theta_{-}},\;\frac{(3+2^{\frac{2}{q}-2})(\theta_{+}-1)}{\theta_{+}}\right\}<t<\frac{4}{3}.
\]
Then, a solution $\hat{x}$ to $(\ref{mod:lq})$ satisfies
\begin{equation}\label{add-1}
\min{\{\|\hat{x}-x_0\|_2,\;\|\hat{x}+x_0\|_2\}}\le c_1\epsilon+c_2\frac{2^{\frac{2}{q}-1}\sigma_{k}(D^*x_0)_q}{k^{\frac{1}{q}-\frac{1}{2}}},
\end{equation}
where
\begin{align*}
c_1&=\frac{(1+2^{\frac{1}{q}-1})\tilde{t}\sqrt{1+\delta_{tk}}}{t-(3+2^{\frac{2}{q}-2}-t)\delta_{tk}},\\
c_2&=\frac{(1+2^{\frac{1}{q}-1})\left\{2^{\frac{1}{q}}\delta_{tk}+\sqrt{\left[t-\big(3+2^{\frac{2}{q}-2}-t\big)\delta_{tk}\right]\delta_{tk}}\right\}}{t-(3+2^{\frac{2}{q}-2}-t)\delta_{tk}}
     +1.
\end{align*}
Here, $\tilde{t}=\max\{t,\sqrt{t}\}$, and $\delta_{tk}$ is a D-RIP constant satisfying
\[
\delta_{tk}\le \max\{1-\theta_{-},\;\theta_{+}-1\}<\frac{t}{3+2^{\frac{2}{q}-2}-t}.
\]
\end{Theorem}

Before proving Theorem \ref{Thm:S-DRIP}, we first introduce a lemma as follows.
\begin{Lemma}\label{Lem:D-RIP}
Let $D$ be a tight frame, $x_0\in D\mathbb{R}_{k}^{N}$. Suppose that a measurement matrix $A\in \mathbb{R}^{m\times n}$ satisfies the D-RIP with a constant $\delta_{tk}<\frac{t}{3+2^{\frac{2}{q}-2}-t}$ for $0<t<\frac{4}{3}$. Then, for any
\begin{align*}
D^*\hat{x}\in\{&D^*x\in \mathbb{R}^N:\|D^*x\|_q^q\le\|D^*x_0\|_q^q,\;\|Ax-Ax_0\|_2\le\epsilon\},
\end{align*}
we have
\begin{equation}\label{Lem:D-RIP1}
\|\hat{x}-x_0\|_2\le c_1\epsilon+c_2\frac{2^{\frac{2}{q}-1}\sigma_{k}(D^*x_0)_q}{k^{\frac{1}{q}-\frac{1}{2}}},
\end{equation}
where
\begin{align*}
c_1&=\frac{(1+2^{\frac{1}{q}-1})\tilde{t}\sqrt{1+\delta_{tk}}}{t-(3+2^{\frac{2}{q}-2}-t)\delta_{tk}},\\
c_2&=\frac{(1+2^{\frac{1}{q}-1})\left\{2^{\frac{1}{q}}\delta_{tk}+\sqrt{\left[t-\big(3+2^{\frac{2}{q}-2}-t\big)\delta_{tk}\right]\delta_{tk}}\right\}}{t-(3+2^{\frac{2}{q}-2}-t)\delta_{tk}}+1,
\end{align*}
and $\tilde{t}=\max\{t,\sqrt{t}\}$.
\end{Lemma}

The derivation of Lemma \ref{Lem:D-RIP} is similar to that of Theorem 2 in \cite{ZL18}. We present the proof of Lemma \ref{Lem:D-RIP} in the appendix.

\begin{proof}[Proof of Theorem~\ref{Thm:S-DRIP}]
Since $\hat{x}$ is a solution to $(\ref{mod:lq})$, we have
\begin{equation}\label{S-DRIP:SD2}
\|D^*\hat{x}\|_q^q\le\|D^*x_0\|_q^q
\end{equation}
and
\begin{equation}\label{S-DRIP:SD3}
\||A\hat{x}|-|Ax_0|\|_{2}^2\le\epsilon^2.
\end{equation}
Define all the rows of matrix $A$ by $\left\{a_i^T,\;i=1,2,\cdots,m\right\}$, and divide the index set $[1:m]$ into two groups:
\[
S=\{i:{\rm{sign}}(\langle a_i,\hat{x}\rangle)={\rm{sign}}(\langle a_i,x_0\rangle)\},
\]
and
\[
S^c=\{i:{\rm{sign}}(\langle a_i,\hat{x}\rangle)=-{\rm{sign}}(\langle a_i,x_0\rangle)\}.
\]
Then, we know either $|S|\ge \frac{m}{2}$ or $|S^c|\ge \frac{m}{2}$.

First, we consider the case of $|S|\ge \frac{m}{2}$. From (\ref{S-DRIP:SD3}), we have
\begin{align}\label{S-DRIP:SD4}
\|A_{S}\hat{x}-A_{S}x_0\|_{2}^2&\le\|A_{S}\hat{x}-A_{S}x_0\|_{2}^2+\|A_{S^c}\hat{x}+A_{S^c}x_0\|_{2}^2\le\epsilon^2.
\end{align}
It follows from (\ref{S-DRIP:SD2}) and (\ref{S-DRIP:SD4}) that
\begin{align}\label{S-DRIP:SD5}
D^*\hat{x}\in\{&D^*x\in\mathbb{R}^N:\|D^*x\|_q^q\le\|D^*x_0\|_q^q,\;\|A_{S}x-A_{S}x_0\|_{2}\le\epsilon\}.
\end{align}
Since $A$ satisfies the S-DRIP of order $tk$ with constants $\theta_{-},\;\theta_{+}\in(0,2)$ satisfying
\[
\max\left\{\frac{(3+2^{\frac{2}{q}-2})(1-\theta_{-})}{2-\theta_{-}},\;\frac{(3+2^{\frac{2}{q}-2})(\theta_{+}-1)}{\theta_{+}}\right\}<t<\frac{4}{3},
\]
$A_S$ satisfies the D-RIP of order $tk$ with
\begin{equation}\label{S-DRIP:SD6}
\delta_{tk}\le\max\{1-\theta_{-},\;\theta_{+}-1\}<\frac{t}{3+2^{\frac{2}{q}-2}-t}.
\end{equation}
According to Lemma~\ref{Lem:D-RIP}, we can obtain
\[
\|\hat{x}-x_0\|_2\le c_1\epsilon+c_2\frac{2^{\frac{2}{q}-1}\sigma_{k}(D^*x_0)_q}{k^{\frac{1}{q}-\frac{1}{2}}},
\]
where $c_1$ and $c_2$ are defined above in Lemma \ref{Lem:D-RIP}.

Similarly, for the case of $|S^c|\ge \frac{m}{2}$, we can obtain
\[
\|\hat{x}+x_0\|_2\le c_1\epsilon+c_2\frac{2^{\frac{2}{q}-1}\sigma_{k}(D^*x_0)_q}{k^{\frac{1}{q}-\frac{1}{2}}}.
\]
\end{proof}

\begin{Remark}
In Theorem \ref{Thm:S-DRIP}, letting $\epsilon=0$, and $D^*{x_0}$ be $k$-sparse, the S-DRIP provides a condition for exact recovery in sparse phase retrieval with redundant dictionary via the $\ell_q$-analysis model (\ref{mod:lq}).
\end{Remark}


\begin{Remark}
When $q=1$, the $\ell_q$-analysis model (\ref{mod:lq}) reduces to the $\ell_1$-analysis model (\ref{mod:l1}), and at this point, our Theorem \ref{Thm:S-DRIP} presents the S-DRIP of order $tk$ with $0<t<\frac{4}{3}$ for stable recovery in sparse phase retrieval with redundant dictionary, which is complementary to Theorem 4.1 in \cite{Gao2017} where the S-DRIP of order $tk$ with $t>1$ was studied.
\end{Remark}

In particular, when $q=1$ and $D=I_n$, the $\ell_q$-analysis model (\ref{mod:lq}) corresponds to the standard sparse phase retrieval via the $\ell_1$ minimization, i.e.,
\begin{equation}\label{mod:l2}
\min\|x\|_{1} \quad{\mbox{subject to}}\quad\||Ax|-|Ax_0|\|_{2}^2\le\epsilon^2.
\end{equation}
Then, our Theorem \ref{Thm:S-DRIP} becomes complementary to Theorem 3.1 in \cite{GWX2016} as well. Specifically, we can obtain the following corollary of the $\ell_1$ minimization for sparse phase retrieval, provided that $A$ satisfies the S-RIP of order $tk$ with $0<t<\frac{4}{3}$.

\begin{Corollary}\label{Cor1}
Suppose that $A\in \mathbb{R}^{m\times n}$ satisfies the S-RIP of order $tk$ with $\theta_{-},\;\theta_{+}\in(0,2)$ satisfying
\[
\max\left\{\frac{4(1-\theta_{-})}{2-\theta_{-}},\;\frac{4(\theta_{+}-1)}{\theta_{+}}\right\}<t<\frac{4}{3}.
\]
Then, any solution $\hat{x}$ to the $\ell_1$ minimization $(\ref{mod:l2})$ satisfies
\begin{equation*}\label{S-DRIP:SD1}
\min{\{\|\hat{x}-x_0\|_2,\;\|\hat{x}+x_0\|_2\}}\le c_3\epsilon+c_4\frac{2\sigma_{k}(x_0)_1}{\sqrt{k}},
\end{equation*}
where
\begin{align*}
c_3=\frac{2\tilde{t}\sqrt{1+\delta_{tk}}}{t-(4-t)\delta_{tk}},\;c_4=\frac{4\delta_{tk}+2\sqrt{\big[t-(4-t)\delta_{tk}\big]\delta_{tk}}}{t-(4-t)\delta_{tk}}+1.
\end{align*}
Here, $\tilde{t}=\max\{t,\sqrt{t}\}$, and $\delta_{tk}$ is a D-RIP constant satisfying
\[
\delta_{tk}\le \max\{1-\theta_{-},\;\theta_{+}-1\}<\frac{t}{4-t}.
\]
\end{Corollary}

\section{Conclusion}\label{sec:S6}
In this paper, we considered the $\ell_q$-analysis model with $0<q\leq1$ for phase retrieval with redundant dictionary, where the original signal of interest is not sparse in the standard orthonormal basis, but is sparse in a redundant dictionary or frame. We first presented two NSP variants, which are necessary and sufficient for exact recovery of real and complex dictionary-$k$-sparse signals in the noiseless scenario. Furthermore, we provided the S-DRIP of order $tk$ with $0<t<\frac{4}{3}$ under which a dictionary-$k$-sparse signal can be stably recovered in the noisy scenario. Of note, in this paper we only focused on the theoretical results on the conditions for exact/stable recovery or dictionary-sparse signals via the $\ell_q$-analysis model. The algorithms of the $\ell_q$-analysis model are worthy of further study. In addition, the S-DRIP based results in this paper only hold in the real number field, and so it is interesting to extend the results to the complex number filed for future research.

\section*{Appendix}
\begin{appendix}
\setcounter{equation}{0}
\renewcommand{\theequation}{A.\arabic{equation}}
\renewcommand{\appendixname}{Appendix~\Alph{section}}

\section{The Proof of Lemma \ref{Lem:D-RIP}}

Inspired by Zhang and Li's work \cite{ZL18}, we provide the proof of Lemma \ref{Lem:D-RIP} as follows.

\begin{proof}[Proof of Lemma \ref{Lem:D-RIP}]
First, we assume that $tk$ is an integer. Let $h=\hat{x}-x_0$, and let $T_0$ and $S_0$ be the index sets of the $k$ largest entries in absolute value of $D^*x_0$ and $D^*h$, respectively. Without loss of generality, we assume that $S_0:=[1,k]$.
Since
\begin{equation*}
\|D^*x_0\|_q^q\ge\|D^*\hat{x}\|_q^q,
\end{equation*}
we have
\begin{eqnarray}
\|D_{T_0}^*x_0\|_q^q+\|D_{T_0^c}^*x_0\|_q^q&=&\|D^*x_0\|_q^q\nonumber\\
&\ge&\|D^*x_0+D^*{h}\|_q^q\nonumber\\
&=&\|D_{T_0}^*x_0+D_{T_0}^*{h}\|_q^q+\|D_{T_0^c}^*x_0+D_{T_0^c}^*{h}\|_q^q\nonumber\\
&\ge&\|D_{T_0}^*x_0\|_q^q-\|D_{T_0}^*{h}\|_q^q+\|D_{T_0^c}^*{h}\|_q^q-\|D_{T_0^c}^*x_0\|_q^q.\label{lem:D1}
\end{eqnarray}
Rewrite the above inequality (\ref{lem:D1}), and we get
\begin{eqnarray}
\|D_{T_0^c}^*{h}\|_q^q&\le&\|D_{T_0}^*{h}\|_q^q+2\|D_{T_0^c}^*{x_0}\|_q^q\nonumber\\
&=&\|D_{T_0}^*{h}\|_q^q+2\sigma_k(D^*{x_0})_q^q.\label{lem:D2}
\end{eqnarray}
Thus, it is obtained that
\begin{eqnarray}
\|D_{S_0^c}^*{h}\|_q^q&\le&\|D_{T_0^c}^*{h}\|_q^q\le\|D_{T_0}^*{h}\|_q^q+2\sigma_k(D^*{x_0})_q^q\nonumber\\
&\le&\|D_{S_0}^*{h}\|_q^q+2\sigma_k(D^*{x_0})_q^q.\label{lem:D3}
\end{eqnarray}
Assume that $\|D_{S_0^c}^*h\|_0=r\ge k$. Choose
\begin{equation}
\alpha^q:=\frac{\|D_{S_0}^*{h}\|_q^q+2\sigma_k(D^*{x_0})_q^q}{k}.\label{lem:Aa1}
\end{equation}
Let $a$ and $b$ be two positive integers satisfying $a+b=tk$ and $b\le a\le k$. Set all the possible index sets $T_i$,\;$S_j\subset S_0=[1:k]$, which satisfy $|T_i|=a\;(i\in I),|S_j|=b\;(j\in J)$ with $|I|=\left(
                                                                                             \begin{array}{c}
                                                                                               k \\
                                                                                               a \\
                                                                                             \end{array}
                                                                                           \right)$,
$|J|=\left(
                                                                                             \begin{array}{c}
                                                                                               k \\
                                                                                               b \\
                                                                                             \end{array}
                                                                                           \right)$,
respectively.
Note that
\[
\|D_{S_0^c}^*{h}\|_q^q\le\|D_{S_0}^*{h}\|_q^q+2\sigma_k(D^*{x_0})_q^q=k\alpha^q=b\cdot \left(\sqrt[q]{\frac{k}{b}}\alpha\right)^q.
\]
Then, we can derive
\begin{equation*}
\|D_{S_0^c}^*{h}\|_{\infty}^q\le\frac{\|D_{S_0}^*{h}\|_q^q}{k}\le\frac{k\alpha^q}{k}=\alpha^q,
\end{equation*}
namely,
\[
\|D_{S_0^c}^*{h}\|_{\infty}\le\alpha\le\sqrt[q]{\frac{k}{b}}\alpha.
\]
Thus, it follows from Lemma \ref{lem2} that $D_{S_0^c}^*{h}$ can be decomposed as
\[
D_{S_0^c}^*{h}=\sum_{l}\lambda_lu_l,
\]
where $u_l$ is $b$-sparse, and
\begin{equation}\label{lem:dd}
\sum_{l}\lambda_l\|u_l\|_2^2\le\frac{k}{b}\alpha^q\|D_{S_0^c}^*{h}\|_{2-q}^{2-q}.
\end{equation}
Similarly, $D_{S_0^c}^*{h}$ can also be represented by
\[
D_{S_0^c}^*{h}=\sum_{l}\lambda_l^\prime v_l,\;D_{S_0^c}^*{h}=\sum_{l}\lambda_l^{\prime\prime} w_l,
\]
where $v_l$ is $a$-sparse, $w_l$ is $(t-1)k$-sparse (for $t>1$), and
\begin{eqnarray}
\sum_{l}\lambda_l^\prime\|v_l\|_2^2\le\frac{k}{a}\alpha^q\|D_{S_0^c}^*{h}\|_{2-q}^{2-q},\;\sum_{l}\lambda_l^{\prime\prime}\|w_l\|_2^2\le\frac{1}{t-1}\alpha^q\|D_{S_0^c}^*{h}\|_{2-q}^{2-q}.\label{lem:D4}
\end{eqnarray}
Note that
\begin{equation}\label{lem:D5}
\|Ah\|_2=\|A\hat{x}-Ax_0\|_2\le\epsilon.
\end{equation}
We have
\begin{eqnarray}\label{lem:D6}
|\langle ADD_{S_0}^*h,Ah\rangle|^2&\le&\|ADD_{S_0}^*h\|_2^2\|Ah\|_2^2\nonumber\\
&\le&\epsilon^2(1+\delta_k)\|DD_{S_0}^*h\|_2^2\nonumber\\
&\le&\epsilon^2(1+\delta_{tk})\|D_{S_0}^*h\|_2^2
\end{eqnarray}
for $t\in[1,\frac{4}{3})$, where we use the fact that if $k\le k_1$, $\delta_k\le\delta_{k_1}$ in the last step;
\begin{eqnarray}\label{lem:D7}
|\langle ADD_{S_0}^*h,Ah\rangle|^2&\le&\epsilon^2(1+\delta_k)\|DD_{S_0}^*h\|_2^2\nonumber\\
&\le&\epsilon^2\left[1+\left(\frac{2}{t}-1\right)\delta_{tk}\right]\|D_{S_0}^*h\|_2^2\nonumber\\
&\le&\frac{(1+\delta_{tk})\epsilon^2}{t}\|D_{S_0}^*h\|_2^2
\end{eqnarray}
for $t\in(0,1)$, where we use Lemma 4.1 proposed in \cite{CZ13} in the last two steps.
Set
\begin{align}
Y_{a,b}&:=\frac{k-b}{a\left(
                                                                                             \begin{array}{c}
                                                                                               k \\
                                                                                               a \\
                                                                                             \end{array}
                                                                                           \right)}
\sum_{i\in I,l}\lambda_l\bigg[a^2\|AD\left(D_{T_i}^*h+\frac{b}{k}u_l\right)\|_2^2+a^2\|D^{\bot}\left(D_{T_i}^*h+\frac{b}{k}u_l\right)\|_2^2\nonumber\\
    &\quad-b^2\|AD\left(D_{T_i}^*h-\frac{a}{k}u_l\right)\|_2^2-b^2\|D^{\bot}\left(D_{T_i}^*h-\frac{a}{k}u_l\right)\|_2^2\bigg]\nonumber\\
    &\quad\quad+\frac{k-a}{b\left(
                                                                                             \begin{array}{c}
                                                                                               k \\
                                                                                               b \\
                                                                                             \end{array}
                                                                                           \right)}
\sum_{j\in J,l}\lambda^{\prime}_l\bigg[b^2\|AD\left(D_{S_j}^*h+\frac{a}{k}v_l\right)\|_2^2+b^2\|D^{\bot}\left(D_{S_j}^*h+\frac{a}{k}v_l\right)\|_2^2\nonumber\\
     &\quad\quad\quad-a^2\|AD\left(D_{S_j}^*h-\frac{b}{k}v_l\right)\|_2^2-a^2\|D^{\bot}\left(D_{S_j}^*h-\frac{b}{k}v_l\right)\|_2^2\bigg].\label{lem:D8}
\end{align}
For $1\le t<\frac{4}{3}$, we can obtain
\begin{align}
&(4-3t)Y_{a,b}-2t^3[ab-(t-1)k^2]\langle ADD_{S_0}^*h,Ah\rangle\nonumber\\
=&C_{t,a,b}\sum_{l}\lambda_l^{\prime\prime}\bigg[\|AD(D_{S_0}^*h+(t-1)w_l)\|_2^2+\|D^{\bot}(D_{S_0}^*h+(t-1)w_l)\|_2^2\nonumber\\
 &\quad-(t-1)^2\|AD(D_{S_0}^*h-w_l)\|_2^2-(t-1)^2\|D^{\bot}(D_{S_0}^*h-w_l)\|_2^2\bigg],\label{lem:D9}
\end{align}
and for $0<t<1$, we can derive
\begin{align}
&tY_{a,b}-2t^2(2-t)ab\langle ADD_{S_0}^*h,Ah\rangle\nonumber\\
=&\frac{C_{t,a,b}}{\left(
                                                                                             \begin{array}{c}
                                                                                               k \\
                                                                                               a \\
                                                                                             \end{array}
                                                                                           \right)
                                                                                           \left(
                                                                                             \begin{array}{c}
                                                                                               k-a \\
                                                                                               b \\
                                                                                             \end{array}
                                                                                           \right)}
\sum_{T_i\cap S_j=\emptyset}\bigg[\|AD\left(D_{T_i}^*h+D_{S_j}^*h\right)\|_2^2+\|D^{\bot}\left(D_{T_i}^*h+D_{S_j}^*h\right)\|_2^2\nonumber\\
&\quad-\frac{1-t}{ab}\|AD\left(bD_{T_i}^*h-aD_{S_j}^*h\right)\|_2^2-\frac{1-t}{ab}\|D^{\bot}\left(bD_{T_i}^*h-aD_{S_j}^*h\right)\|_2^2\bigg],\label{lem:D10}
\end{align}
where it is due to the fact that $\langle D^{\bot}D_{S_0}^*h,D^{\bot}D^*h\rangle=0$ and $C_{t,a,b}:=(a-b)^2-2(2-t)ab$. The proofs of the above two inequalities are similar to those of (14) and (15) in \cite{ZL18}, and therefore, we leave out all the details for simplicity. When $tk$ is even, let $a=b=\frac{tk}{2}$; when $tk$ is odd, let $a=b+1=\frac{tk+1}{2}$. We can easily prove $C_{t,a,b}<0$ for both cases.

Since $D_{T_i}^*,\;v_l$ are $a$-sparse, $D_{S_j}^*,\;u_l$ are $b$-sparse with $a+b=tk$, we can apply the D-RIP of order $tk$ to $Y_{a,b}$, which yields
\begin{align}
&Y_{a,b}\nonumber\\
\ge&\frac{k-b}{a\left(
                                                                                             \begin{array}{c}
                                                                                               k \\
                                                                                               a \\
                                                                                             \end{array}
                                                                                           \right)}
\sum_{i\in I,l}\lambda_l\bigg[a^2(1-\delta_{tk})\|D_{T_i}^*h+\frac{b}{k}u_l\|_2^2-b^2(1+\delta_{tk})\|D_{T_i}^*h-\frac{a}{k}u_l\|_2^2\bigg]\nonumber\\
&\quad+\frac{k-a}{b\left(
                                                                                             \begin{array}{c}
                                                                                               k \\
                                                                                               b \\
                                                                                             \end{array}
                                                                                           \right)}
\sum_{j\in J,l}\lambda^{\prime}_l\bigg[b^2(1-\delta_{tk})\|D_{S_j}^*h+\frac{a}{k}v_l\|_2^2-a^2(1+\delta_{tk})\|D_{S_j}^*h-\frac{b}{k}v_l\|_2^2\bigg]\nonumber\\
=&\frac{k-b}{a\left(
                                                                                             \begin{array}{c}
                                                                                               k \\
                                                                                               a \\
                                                                                             \end{array}
                                                                                           \right)}
\sum_{i\in I,l}\lambda_l\bigg[a^2(1-\delta_{tk})\bigg(\|D_{T_i}^*h\|_2^2+\frac{b^2}{k^2}\|u_l\|_2^2\bigg)-b^2(1+\delta_{tk})\bigg(\|D_{T_i}^*h\|_2^2+\frac{a^2}{k^2}\|u_l\|_2^2\bigg)\bigg]\nonumber\\
&\quad+\frac{k-a}{b\left(
                                                                                             \begin{array}{c}
                                                                                               k \\
                                                                                               b \\
                                                                                             \end{array}
                                                                                           \right)}
\sum_{j\in J,l}\lambda^{\prime}_l\bigg[b^2(1-\delta_{tk})\bigg(\|D_{S_j}^*h\|_2^2+\frac{a^2}{k^2}\|v_l\|_2^2\bigg)\nonumber\\
&\quad\quad-a^2(1+\delta_{tk})\bigg(\|D_{S_j}^*h\|_2^2+\frac{b^2}{k^2}\|v_l\|_2^2\bigg)\bigg]\nonumber\\
=&(a^2-b^2)\Bigg[\frac{k-b}{a\left(
                                                                                             \begin{array}{c}
                                                                                               k \\
                                                                                               a \\
                                                                                             \end{array}
                                                                                           \right)}
\sum_{i\in I}\|D_{T_i}^*h\|_2^2-\frac{k-a}{b\left(
                                                                                             \begin{array}{c}
                                                                                               k \\
                                                                                               b \\
                                                                                             \end{array}
                                                                                           \right)}
\sum_{j\in J}\|D_{S_j}^*h\|_2^2\Bigg]\nonumber\\
&\quad-(a^2+b^2)\delta_{tk}\Bigg[\frac{k-b}{a\left(
                                                                                             \begin{array}{c}
                                                                                               k \\
                                                                                               a \\
                                                                                             \end{array}
                                                                                           \right)}
\sum_{i\in I}\|D_{T_i}^*h\|_2^2
+\frac{k-a}{b\left(
                                                                                             \begin{array}{c}
                                                                                               k \\
                                                                                               b \\
                                                                                             \end{array}
                                                                                           \right)}
\sum_{j\in J}\|D_{S_j}^*h\|_2^2\Bigg]\nonumber\\
&\quad\quad-\frac{2\delta_{tk}ab}{k^2}\bigg[b(k-b)\sum_{l}\lambda_l\|u_l\|_2^2+a(k-a)\sum_{l}\lambda_l^{\prime}\|v_l\|_2^2\bigg]\nonumber\\
\ge&(a^2-b^2)\Bigg[\frac{k-b}{a\left(
                                                                                             \begin{array}{c}
                                                                                               k \\
                                                                                               a \\
                                                                                             \end{array}
                                                                                           \right)}
\left(
                                                                                             \begin{array}{c}
                                                                                               k-1 \\
                                                                                               a-1 \\
                                                                                             \end{array}
                                                                                           \right)
\|D_{S_0}^*h\|_2^2-\frac{k-a}{b\left(
                                                                                             \begin{array}{c}
                                                                                               k \\
                                                                                               b \\
                                                                                             \end{array}
                                                                                           \right)}
                                                                                           \left(
                                                                                             \begin{array}{c}
                                                                                               k-1 \\
                                                                                               b-1 \\
                                                                                             \end{array}
                                                                                           \right)
\|D_{S_0}^*h\|_2^2\Bigg]\nonumber\\
&\quad-(a^2+b^2)\delta_{tk}\Bigg[\frac{k-b}{a\left(
                                                                                             \begin{array}{c}
                                                                                               k \\
                                                                                               a \\
                                                                                             \end{array}
                                                                                           \right)}
                                                                                           \left(
                                                                                             \begin{array}{c}
                                                                                               k-1 \\
                                                                                               a-1 \\
                                                                                             \end{array}
                                                                                           \right)
\|D_{S_0}^*h\|_2^2-\frac{k-a}{b\left(
                                                                                             \begin{array}{c}
                                                                                               k \\
                                                                                               b \\
                                                                                             \end{array}
                                                                                           \right)}
                                                                                           \left(
                                                                                             \begin{array}{c}
                                                                                               k-1 \\
                                                                                               b-1 \\
                                                                                             \end{array}
                                                                                           \right)
\|D_{S_0}^*h\|_2^2\Bigg]\nonumber\\
&\quad\quad-\frac{2\delta_{tk}ab}{k^2}\bigg[b(k-b)\frac{k}{b}\alpha^q\|D_{S_0^c}^*h\|_{2-q}^{2-q}+a(k-a)\frac{k}{a}\alpha^q\|D_{S_0^c}^*h\|_{2-q}^{2-q}\bigg]\nonumber\\
=&[t(a-b)^2-(2-t)(a^2+b^2)\delta_{tk}]\|D_{S_0}^*h\|_2^2-2(2-t)\delta_{tk}ab\alpha^q\|D_{S_0^c}^*h\|_{2-q}^{2-q},\label{lem:D11}
\end{align}
where we use (\ref{lem1:1}), (\ref{lem:dd}) and (\ref{lem:D4}) in the penultimate step.
By the H\"{o}lder's inequality, (\ref{lem:D3}) and (\ref{lem:Aa1}), we have
\begin{eqnarray}
\alpha^q\|D_{S_0^c}^*h\|_{2-q}^{2-q}
&=&\alpha^q\sum_{i\in S_0^c}|D^*h(i)|^{\frac{2(2-2q)}{2-q}}|D^*h(i)|^{\frac{q^2}{2-q}}\nonumber\\
&\le&\alpha^q\left(\sum_{i\in S_0^c}|D^*h(i)|^{2}\right)^{\frac{2-2q}{2-q}}\left(\sum_{i\in S_0^c}|D^*h(i)|^q\right)^{\frac{q}{2-q}}\nonumber\\
&=&\alpha^q\left(\|D_{S_0^c}^*h\|_{2}^2\right)^{\frac{2-2q}{2-q}}\left(\|D_{S_0^c}^*h\|_{q}^q\right)^{\frac{q}{2-q}}\nonumber\\
&\le&\alpha^q\left(\|D_{S_0^c}^*h\|_{2}^2\right)^{\frac{2-2q}{2-q}}(k\alpha^q)^{\frac{q}{2-q}}\nonumber\\
&=&\left(\|D_{S_0^c}^*h\|_{2}^2\right)^{\frac{2-2q}{2-q}}\left(\frac{k\alpha^q}{k^{1-\frac{q}{2}}}\right)^{\frac{2}{2-q}}\nonumber\\
&=&\left(\|D_{S_0^c}^*h\|_{2}^2\right)^{\frac{2-2q}{2-q}}\left(\frac{\|D_{S_0}^*{h}\|_q^q+2\sigma_k(D^*{x_0})_q^q}{k^{1-\frac{q}{2}}}\right)^{\frac{2}{2-q}}\nonumber\\
&\le&\left(\|D_{S_0^c}^*h\|_{2}^2\right)^{\frac{2-2q}{2-q}}\left(\|D_{S_0}^*{h}\|_2^q+\frac{2\sigma_k(D^*{x_0})_q^q}{k^{1-\frac{q}{2}}}\right)^{\frac{2}{2-q}}.\label{lem:D12}
\end{eqnarray}
Combining (\ref{lem:D3}) and Lemma~\ref{Lem:Add1}, we derive
\begin{eqnarray}
\|D_{S_0^c}^*h\|_{2}^2&\le& k\left[\left(\frac{\|D_{S_0}^*{h}\|_2^2}{k}\right)^{\frac{q}{2}}+\frac{2\sigma_k(D^*{x_0})_q^q}{k}\right]^{\frac{2}{q}}\nonumber\\
&=&\left(\|D_{S_0}^*{h}\|_2^q+\frac{2\sigma_k(D^*{x_0})_q^q}{k^{1-\frac{q}{2}}}\right)^{\frac{2}{q}}.\label{lem:A12}
\end{eqnarray}
Substituting (\ref{lem:A12}) into (\ref{lem:D12}), we obtain
\begin{eqnarray}
\alpha^q\|D_{S_0^c}^*h\|_{2-q}^{2-q}
&\le&\left(\|D_{S_0}^*{h}\|_2^q+\frac{2\sigma_k(D^*{x_0})_q^q}{k^{1-\frac{q}{2}}}\right)^{\frac{2}{q}}\nonumber\\
&\le&\left[2^{\frac{1}{q}-1}\bigg(\|D_{S_0}^*h\|_2+\frac{2^{\frac{1}{q}}\sigma_k(D^*x_0)_q}{k^{\frac{1}{q}-\frac{1}{2}}}\bigg)\right]^2\nonumber\\
&=&2^{\frac{2}{q}-2}\|D_{S_0}^*h\|_2^2+2^{\frac{3}{q}-1}\|D_{S_0}^*h\|_2R+2^{\frac{4}{q}-2}R^2,\label{lem:A13}
\end{eqnarray}
where $$R:=\frac{\sigma_k(D^*{x_0})_q}{k^{\frac{1}{q}-\frac{1}{2}}}.$$
For brevity, set
\begin{eqnarray*}
\mathcal{A}&:=&(4-3t)Y_{a,b}-2t^3\big[ab-(t-1)k^2\big]\langle ADD_{S_0}^*h,Ah\rangle,\nonumber\\
\mathcal{B}&:=&C_{t,a,b}\sum_{l}\lambda_l^{\prime\prime}\bigg[\|AD(D_{S_0}^*h+(t-1)w_l)\|_2^2+\|D^{\bot}(D_{S_0}^*h+(t-1)w_l)\|_2^2{}\nonumber\\
{}&&\quad-(t-1)^2\|AD(D_{S_0}^*h-w_l)\|_2^2-(t-1)^2\|D^{\bot}(D_{S_0}^*h-w_l)\|_2^2\bigg],\nonumber\\
\mathcal{C}&:=&tY_{a,b}-2t^2(2-t)ab\langle ADD_{S_0}^*h,Ah\rangle,\nonumber\\
\mathcal{D}&:=&\frac{C_{t,a,b}}{\left(
                                                                                             \begin{array}{c}
                                                                                               k \\
                                                                                               a \\
                                                                                             \end{array}
                                                                                           \right)
                                                                                           \left(
                                                                                             \begin{array}{c}
                                                                                               k-a \\
                                                                                               b \\
                                                                                             \end{array}
                                                                                           \right)}
\sum_{T_i\cap S_j=\emptyset}\bigg[\|AD\left(D_{T_i}^*h+D_{S_j}^*h\right)\|_2^2+\|D^{\bot}\left(D_{T_i}^*h+D_{S_j}^*h\right)\|_2^2{}\nonumber\\
{}&&-\frac{1-t}{ab}\|AD\left(bD_{T_i}^*h-aD_{S_j}^*h\right)\|_2^2-\frac{1-t}{ab}\|D^{\bot}\left(bD_{T_i}^*h-aD_{S_j}^*h\right)\|_2^2\bigg].\nonumber\\
\end{eqnarray*}

When $t\in[1,\frac{4}{3})$, combining (\ref{lem:D6}) and (\ref{lem:D11}), we get
\begin{align}
\mathcal{A}\ge&(4-3t)\Big\{\big[t(a-b)^2-(2-t)(a^2+b^2)\delta_{tk}\big]\|D_{S_0}^*h\|_2^2-2(2-t)\delta_{tk}ab\alpha^q\|D_{S_0^c}^*h\|_{2-q}^{2-q}\Big\}{}\nonumber\\
&-2t^3\big[ab-(t-1)k^2\big]\cdot\epsilon\sqrt{1+\delta_{tk}}\|D_{S_0}^*h\|_2.\label{lem:D15}
\end{align}
Since $D_{S_0}^*h$ is $k$-sparse and $w_l$ is $(t-1)k$-sparse, applying the D-RIP of order $tk$, we obtain
\begin{eqnarray}
\mathcal{B}&\le& C_{t,a,b}\sum_{l}\lambda_l^{\prime\prime}\Big[(1-\delta_{tk})\|D_{S_0}^*h+(t-1)w_l\|_2^2-(t-1)^2(1+\delta_{tk})\|D_{S_0}^*h-w_l\|_2^2\Big]\nonumber\\
&\le&C_{t,a,b}\Big\{\big[1-\delta_{tk}-(t-1)^2(1+\delta_{tk})\big]\|D_{S_0}^*h\|_2^2-2(t-1)^2\delta_{tk}\sum_{l}\lambda_l^{\prime\prime}\|w_l\|_2^2\Big\}\nonumber\\
&\le&C_{t,a,b}\Big\{\big[1-\delta_{tk}-(t-1)^2(1+\delta_{tk})\big]\|D_{S_0}^*h\|_2^2-2(t-1)\delta_{tk}\alpha^q\|D_{S_0^c}^*h\|_{2-q}^{2-q}\Big\},\label{lem:D14}
\end{eqnarray}
where we use (\ref{lem:D4}) in the last step.
Thus, combining (\ref{lem:A13}), (\ref{lem:D15}) and (\ref{lem:D14}), we have
\begin{eqnarray*}
0&\ge&(4-3t)\Big\{\big[t(a-b)^2-(2-t)(a^2+b^2)\delta_{tk}\big]\|D_{S_0}^*h\|_2^2-2(2-t)\delta_{tk}ab\alpha^q\|D_{S_0^c}^*h\|_{2-q}^{2-q}\Big\}{}\nonumber\\
{}&&-2t^3[ab-(t-1)k^2\big]\epsilon\sqrt{1+\delta_{tk}}\|D_{S_0}^*h\|_2{}\nonumber\\
{}&&\quad-C_{t,a,b}\Big\{\big[1-\delta_{tk}-(t-1)^2(1+\delta_{tk})\big]\|D_{S_0}^*h\|_2^2
-2(t-1)\delta_{tk}\alpha^q\|D_{S_0^c}^*h\|_{2-q}^{2-q}\Big\}\nonumber\\
&=&2t^2[ab-(t-1)k^2\big]\left\{\big[t-(3-t)\delta_{tk}\big]\|D_{S_0}^*h\|_2^2-\delta_{tk}\alpha^q\|D_{S_0^c}^*h\|_{2-q}^{2-q}\right\}{}\nonumber\\
{}&&-2t^3\big[ab-(t-1)k^2\big]\epsilon\sqrt{1+\delta_{tk}}\|D_{S_0}^*h\|_2\nonumber\\
&\ge&2t^2\big[ab-(t-1)k^2\big]\left\{\big[t-(3-t)\delta_{tk}\big]\|D_{S_0}^*h\|_2^2-\delta_{tk}\Big(2^{\frac{2}{q}-2}\|D_{S_0}^*h\|_2^2\right.{}\nonumber\\
{}&&\left.+2^{\frac{3}{q}-1}\|D_{S_0}^*h\|_2R+2^{\frac{4}{q}-2}R^2\Big)\right\}
-2t^3\big[ab-(t-1)k^2\big]\epsilon\sqrt{1+\delta_{tk}}\|D_{S_0}^*h\|_2\nonumber\\
&=&2t^2\big[ab-(t-1)k^2\big]\left\{\big[t-(3+2^{\frac{2}{q}-2}-t)\delta_{tk}\big]\|D_{S_0}^*h\|_2^2-\Big(2^{\frac{3}{q}-1}\delta_{tk}R
\right.{}\nonumber\\
{}&&\left.+t\epsilon\sqrt{1+\delta_{tk}}\Big)\|D_{S_0}^*h\|_2-2^{\frac{4}{q}-2}\delta_{tk}R^2\right\},
\end{eqnarray*}
which is a second-order inequality for $\|D_{S_0}^*h\|_2$.
Notice that $ab\ge\frac{(tk)^2-1}{4}=\frac{(2-t)^2k^2}{4}-\frac{1}{4}+(t-1)k^2>(t-1)k^2$. Since $\delta_{tk}<\frac{t}{3+2^{\frac{2}{q}-2}-t}$, we can derive
\begin{align}
&\|D_{S_0}^*h\|_2\nonumber\\
\le&\frac{2^{\frac{3}{q}-1}\delta_{tk}R+t\epsilon\sqrt{1+\delta_{tk}}}{2\Big[t-\big(3+2^{\frac{2}{q}-2}-t\big)\delta_{tk}\Big]}\nonumber\\
&\quad+\frac{\sqrt{\Big(2^{\frac{3}{q}-1}\delta_{tk}R+t\epsilon\sqrt{1+\delta_{tk}}\Big)^2
   +4\Big[t-\big(3+2^{\frac{2}{q}-2}-t\big)\delta_{tk}\Big]2^{\frac{4}{q}-2}\delta_{tk}R^2}}{2\Big[t-\big(3+2^{\frac{2}{q}-2}-t\big)\delta_{tk}\Big]}\nonumber\\
\le&\frac{t\sqrt{1+\delta_{tk}}}{t-\big(3+2^{\frac{2}{q}-2}-t\big)\delta_{tk}}\epsilon+\frac{2^{\frac{3}{q}-1}\delta_{tk}+2^{\frac{2}{q}-1}
   \sqrt{\Big[t-\big(3+2^{\frac{2}{q}-2}-t\big)\delta_{tk}\Big]\delta_{tk}}}{t-\big(3+2^{\frac{2}{q}-2}-t\big)\delta_{tk}}R.
   \label{lem:D17}
\end{align}

When $t\in(0,1)$, combining (\ref{lem:D7}) and (\ref{lem:D11}), we obtain
\begin{eqnarray}\label{lem:D19}
\mathcal{C}&\ge& t\left\{\big[t(a-b)^2-(2-t)(a^2+b^2)\delta_{tk}\big]\|D_{S_0}^*h\|_2^2-2(2-t)\delta_{tk}ab\alpha^q\|D_{S_0^c}^*h\|_{2-q}^{2-q}\right\}{}\nonumber\\
{}&&-2t(2-t)ab\epsilon\sqrt{(1+\delta_{tk})t}\|D_{S_0}^*h\|_{2}.
\end{eqnarray}
Since $D_{T_i}^*h$ is $a$-sparse and $D_{S_j}^*h$ is $b$-sparse, applying the D-RIP of order $tk$ and (\ref{lem1:1}), we have
\begin{align}
\mathcal{D}
\le&\frac{C_{t,a,b}}{\left(
                                                                                             \begin{array}{c}
                                                                                               k \\
                                                                                               a \\
                                                                                             \end{array}
                                                                                           \right)
                                                                                           \left(
                                                                                             \begin{array}{c}
                                                                                               k-a \\
                                                                                               b \\
                                                                                             \end{array}
                                                                                           \right)}
\sum_{T_i\cap S_j=\emptyset}\bigg[(1-\delta_{tk})\|D_{T_i}^*h+D_{S_j}^*h\|_2^2\nonumber\\
&\quad-\frac{1-t}{ab}(1+\delta_{tk})\|bD_{T_i}^*h-aD_{S_j}^*h\|_2^2\bigg]\nonumber\\
=&\frac{C_{t,a,b}}{\left(
                                                                                             \begin{array}{c}
                                                                                               k \\
                                                                                               a \\
                                                                                             \end{array}
                                                                                           \right)
                                                                                           \left(
                                                                                             \begin{array}{c}
                                                                                               k-a \\
                                                                                               b \\
                                                                                             \end{array}
                                                                                           \right)}
\bigg[(1-\delta_{tk})\sum_{T_i\cap S_j=\emptyset}(\|D_{T_i}^*h\|_2^2+\|D_{S_j}^*h\|_2^2)\nonumber\\
&\quad-\frac{1-t}{ab}(1+\delta_{tk})\sum_{T_i\cap S_j=\emptyset}(b^2\|D_{T_i}^*h\|_2^2+a^2\|D_{S_j}^*h\|_2^2)\bigg]\nonumber\\
=&\frac{C_{t,a,b}}{\left(
                                                                                             \begin{array}{c}
                                                                                               k \\
                                                                                               a \\
                                                                                             \end{array}
                                                                                           \right)
                                                                                           \left(
                                                                                             \begin{array}{c}
                                                                                               k-a \\
                                                                                               b \\
                                                                                             \end{array}
                                                                                           \right)}
\Bigg\{(1-\delta_{tk})\bigg[\left(
                                                                                             \begin{array}{c}
                                                                                               k-a \\
                                                                                               b \\
                                                                                             \end{array}
                                                                                           \right)
                                                                                           \left(
                                                                                             \begin{array}{c}
                                                                                               k-1 \\
                                                                                               a-1 \\
                                                                                             \end{array}
                                                                                           \right)
\|D_{S_0}^*h\|_2^2\nonumber\\
&\quad+\left(
                                                                                             \begin{array}{c}
                                                                                              k-b \\
                                                                                               a \\
                                                                                             \end{array}
                                                                                           \right)
                                                                                           \left(
                                                                                             \begin{array}{c}
                                                                                               k-1 \\
                                                                                               b-1 \\
                                                                                             \end{array}
                                                                                           \right)
\|D_{S_0}^*h\|_2^2\bigg]-\frac{1-t}{ab}(1+\delta_{tk})\bigg[b^2\left(
                                                                                             \begin{array}{c}
                                                                                               k-a \\
                                                                                               b \\
                                                                                             \end{array}
                                                                                           \right)
                                                                                           \left(
                                                                                             \begin{array}{c}
                                                                                               k-1 \\
                                                                                               a-1 \\
                                                                                             \end{array}
                                                                                           \right)\nonumber\\
&\quad\quad\times\|D_{S_0}^*h\|_2^2
+a^2\left(
                                                                                             \begin{array}{c}
                                                                                               k-b \\
                                                                                               a \\
                                                                                             \end{array}
                                                                                           \right)
                                                                                           \left(
                                                                                             \begin{array}{c}
                                                                                               k-1 \\
                                                                                               b-1 \\
                                                                                             \end{array}
                                                                                           \right)
\|D_{S_0}^*h\|_2^2\bigg]\Bigg\}\nonumber\\
\le& C_{t,a,b}t\big[t-(2-t)\delta_{tk}\big]\|D_{S_0}^*h\|_2^2.\label{lem:D18}
\end{align}
Therefore, combining (\ref{lem:A13}), (\ref{lem:D19}), and (\ref{lem:D18}), we have
\begin{eqnarray}
0&\ge&t\left\{\big[t(a-b)^2-(2-t)(a^2+b^2)\delta_{tk}\big]\|D_{S_0}^*h\|_2^2-2(2-t)\delta_{tk}ab\alpha^q\|D_{S_0^c}^*h\|_{2-q}^{2-q}\right\}{}\nonumber\\
{}&&-2t(2-t)ab\epsilon\sqrt{(1+\delta_{tk})t}\|D_{S_0}^*h\|_{2}-C_{t,a,b}t\big[t-(2-t)\delta_{tk}\big]\|D_{S_0}^*h\|_2^2\nonumber\\
&=&2t(2-t)ab\left\{\big[t-(3-t)\delta_{tk}\big]\|D_{S_0}^*h\|_2^2-\delta_{tk}\alpha^q\|D_{S_0^c}^*h\|_{2-q}^{2-q}-\epsilon\sqrt{(1+\delta_{tk})t}\|D_{S_0}^*h\|_{2}\right\}\nonumber\\
&\ge&2t(2-t)ab\Big\{\big[t-(3-t)\delta_{tk}\big]\|D_{S_0}^*h\|_2^2-\delta_{tk}\Big(2^{\frac{2}{q}-2}\|D_{S_0}^*h\|_2^2+2^{\frac{3}{q}-1}\|D_{S_0}^*h\|_2R{}\nonumber\\
{}&&+2^{\frac{4}{q}-2}R^2\Big)-\epsilon\sqrt{(1+\delta_{tk})t}\|D_{S_0}^*h\|_{2}\Big\}\nonumber\\
&=&2t(2-t)ab\left\{\big[t-(3+2^{\frac{2}{q}-2}-t)\delta_{tk}\big]\|D_{S_0}^*h\|_2^2-\big[\epsilon\sqrt{(1+\delta_{tk})t}+2^{\frac{3}{q}-1}\delta_{tk}R\big]\|D_{S_0}^*h\|_2\right.{}\nonumber\\
{}&&\left.-\delta_{tk}2^{\frac{4}{q}-2}R^2\right\}.\label{lem:D20}
\end{eqnarray}
Since $\delta_{tk}<\frac{t}{3+2^{\frac{2}{q}-2}-t}$, we have
\begin{align}
&\|D_{S_0}^*h\|_2\nonumber\\
\le&\frac{\epsilon\sqrt{(1+\delta_{tk})t}+2^{\frac{3}{q}-1}\delta_{tk}R}{2\Big[t-\big(3+2^{\frac{2}{q}-2}-t\big)\delta_{tk}\Big]}\nonumber\\
&\quad+\frac{\sqrt{\Big[\epsilon\sqrt{(1+\delta_{tk})t}+2^{\frac{3}{q}-1}\delta_{tk}R\Big]^2
+4\Big[t-\big(3+2^{\frac{2}{q}-2}-t\big)\delta_{tk}\Big]\delta_{tk}2^{\frac{4}{q}-2}R^2}}{2\Big[t-\big(3+2^{\frac{2}{q}-2}-t\big)\delta_{tk}\Big]}\nonumber\\
\le&\frac{\sqrt{(1+\delta_{tk})t}}{t-\big(3+2^{\frac{2}{q}-2}-t\big)\delta_{tk}}\epsilon+\frac{2^{\frac{3}{q}-1}\delta_{tk}
+2^{\frac{2}{q}-1}\sqrt{\Big[t-(3+2^{\frac{2}{q}-2}-t)\delta_{tk}\Big]\delta_{tk}}}{t-\big(3+2^{\frac{2}{q}-2}-t\big)\delta_{tk}}R.\label{lem:D21}
\end{align}
Combining (\ref{lem:D17}) and (\ref{lem:D21}), we know that
\begin{equation}\label{lem:D22}
\|D_{S_0}^*h\|_2
\le\frac{\tilde{t}\sqrt{1+\delta_{tk}}}{t-\big(3+2^{\frac{2}{q}-2}-t\big)\delta_{tk}}\epsilon+\frac{2^{\frac{3}{q}-1}\delta_{tk}
+2^{\frac{2}{q}-1}\sqrt{\Big[t-\big(3+2^{\frac{2}{q}-2}-t\big)\delta_{tk}\Big]\delta_{tk}}}{t-\big(3+2^{\frac{2}{q}-2}-t\big)\delta_{tk}}R,
\end{equation}
where $\tilde{t}=\max\{t,\sqrt{t}\}$.
Using (\ref{lem:A12}) and (\ref{lem:D22}), we obtain
\begin{align}
&\|h\|_2=\|D^*h\|_2\le\|D_{S_0}^*h\|_2+\|D_{S_0^c}^*h\|_2\nonumber\\
\le&(1+2^{\frac{1}{q}-1})\|D_{S_0}^*h\|_2+2^{\frac{2}{q}-1}R\nonumber\\
\le&\frac{(1+2^{\frac{1}{q}-1})\tilde{t}\sqrt{1+\delta_{tk}}}{t-\big(3+2^{\frac{2}{q}-2}-t\big)\delta_{tk}}\epsilon
   +\left(\frac{(1+2^{\frac{1}{q}-1})\left\{2^{\frac{1}{q}}\delta_{tk}+\sqrt{\left[t-\big(3+2^{\frac{2}{q}-2}-t\big)\delta_{tk}\right]\delta_{tk}}\right\}}{t-\big(3+2^{\frac{2}{q}-2}-t\big)\delta_{tk}}
     +1\right)\nonumber\\
&\quad\times\frac{2^{\frac{2}{q}-1}\sigma_{k}(D^*x_0)_q}{k^{\frac{1}{q}-\frac{1}{2}}}.
\end{align}

If $tk$ is not an integer, set $t^{\prime}:=\frac{\lceil tk\rceil}{k}$, and then we can prove the above result by working on $\delta_{t^{\prime}k}$.

\end{proof}
\end{appendix}

\end{document}